\documentclass[lettersize,journal]{IEEEtran}
\usepackage{amsmath,amsfonts}
\usepackage{algorithmic}
\usepackage{algorithm}
\usepackage{array}
\usepackage[caption=false,font=normalsize,labelfont=sf,textfont=sf]{subfig}
\usepackage{textcomp}
\usepackage{stfloats}
\usepackage{url}
\usepackage{verbatim}
\usepackage{graphicx}
\usepackage{cite}
\usepackage{float}
\usepackage{caption}
\usepackage{cite}
\usepackage{amsmath,amssymb,amsfonts}
\usepackage{algorithmic}
\usepackage{graphicx}
\usepackage{textcomp}
\usepackage{xcolor}
\usepackage{amssymb}
\usepackage{amsmath}
\usepackage{amsthm}
\usepackage{bbm}
\usepackage{enumitem}
\usepackage[justification=raggedright, singlelinecheck=false]{caption}  

\def\BibTeX{{\rm B\kern-.05em{\sc i\kern-.025em b}\kern-.08em
    T\kern-.1667em\lower.7ex\hbox{E}\kern-.125emX}}

\newtheorem{theorem}{Theorem}
\newtheorem{lemma}{Lemma}
\newtheorem{definition}{Definition}
\newtheorem{fact}{Fact}

\newtheorem{proposition}{Proposition}
\newtheorem{assumption}{Assumption}

\newcommand{\ie}{{\em i.e.}}

\newenvironment{proofof}[1]
  {\begin{proof}[\normalfont\textbf{Proof of #1:}]}
  {\end{proof}}

\begin{document}

\title{On Resolving Non-Preemptivity in Multitask Scheduling: An Optimal Algorithm in Deterministic and Stochastic Worlds}

\author{Wenxin Li ~\IEEEmembership{}
\thanks{An early version of this paper was accepted by ACM SIGMETRICS 2022 as a poster paper \cite{li2023performance}.

W. Li is with Department of Electrical and Computer Engineering, The Ohio State University, Columbus, USA. (e-mail: wenxinliwx.1@gmail.com)}
}

\markboth{IEEE Transactions on Networking}%
{Shell \MakeLowercase{\textit{et al.}}: A Sample Article Using IEEEtran.cls for IEEE Journals}


\maketitle

\begin{abstract}
The efficient scheduling of multi-task jobs across multiprocessor systems has become increasingly critical with the rapid expansion of computational systems. This challenge, known as \emph{Multiprocessor Multitask Scheduling (MPMS)}, is essential for optimizing the performance and scalability of applications in fields such as cloud computing and deep learning. In this paper, we study the MPMS problem under both deterministic and stochastic models, where each job is composed of multiple tasks and can only be completed when all its tasks are finished. We introduce $\mathsf{NP}$-$\mathsf{SRPT}$, a non-preemptive variant of the Shortest Remaining Processing Time (SRPT) algorithm, designed to accommodate scenarios with non-preemptive tasks. Our algorithm achieves a competitive ratio of $\ln \alpha + \beta + 1$ for minimizing response time, where $\alpha$ represents the ratio of the largest to the smallest job workload, and $\beta$ captures the ratio of the largest non-preemptive task workload to the smallest job workload. We further establish that this competitive ratio is order-optimal when the number of processors is fixed. For stochastic systems modeled as $\mathsf{M}$/$\mathsf{G}$/$\mathsf{N}$ queues, where job arrivals follow a Poisson process and task workloads are drawn from a general distribution, we prove that $\mathsf{NP}$-$\mathsf{SRPT}$ achieves asymptotically optimal mean response time as the traffic intensity $\rho$ approaches $1$, assuming the task size distribution has finite support. Moreover, the asymptotic optimality extends to cases with infinite task size distributions under mild probabilistic assumptions, including the standard $\mathsf{M}$/$\mathsf{M}$/$\mathsf{N}$ model. {\color{black}{Finally, we extend the analysis to the practical setting of unknown job sizes, proving that non-preemptive adaptations of the $\mathsf{M\text{-}Gittins}$ and $\mathsf{M\text{-}SERPT}$ policies achieve asymptotic optimality and near-optimality, respectively, for a broad class of job size distributions.}} Experimental results validate the effectiveness of $\mathsf{NP}$-$\mathsf{SRPT}$, demonstrating its asymptotic optimality in both theoretical and practical settings.
\end{abstract}


\section{Introduction}

Scheduling is fundamentally about the optimal allocation of resources over time to perform a collection of jobs. With widespread applications in various fields, scheduling jobs to minimize the total response time (also known as flow time \cite{leonardi1997approximating}, sojourn time \cite{DBLP:journals/orl/Bansal05} and delay \cite{wang2019delay}) is a fundamental problem in computer science and operation research that has been extensively studied. As an important metric measuring the quality of a scheduler, response time, is formally defined as the difference between job completion time \cite{li2023work, li2020asymptotic} and releasing date, and characterizes the amount of time that the job spends in the system. 

Optimizing the response time of single-task jobs has been considered both in offline and online scenarios. If preemption is allowed, the \emph{Shortest Remaining Processing Time} (SRPT) discipline is shown to be optimal in single machine environment. Many generalizations of this basic formulation become NP-hard, for example, minimizing the total response time in non-preemptive single machine model and preemptive model with two machines~\cite{leonardi1997approximating}. When jobs arrive online, no information about jobs is known to the algorithm in advance, several algorithms with logarithmic competitive ratios are proposed in various settings~\cite{ azar2018improved,leonardi1997approximating}. On the other hand, while SRPT minimizes the mean response time sample-path wise, it requires the knowledge of remaining job service time. Gittins proved that the Gittins index policy minimizes the mean response time in an $\mathsf{M}$/$\mathsf{G}$/$1$ queue, which only requires the access to the information about job size distribution \cite{gittins2011multi}.

\begin{figure*}
    \centering
    \includegraphics[width=\linewidth]{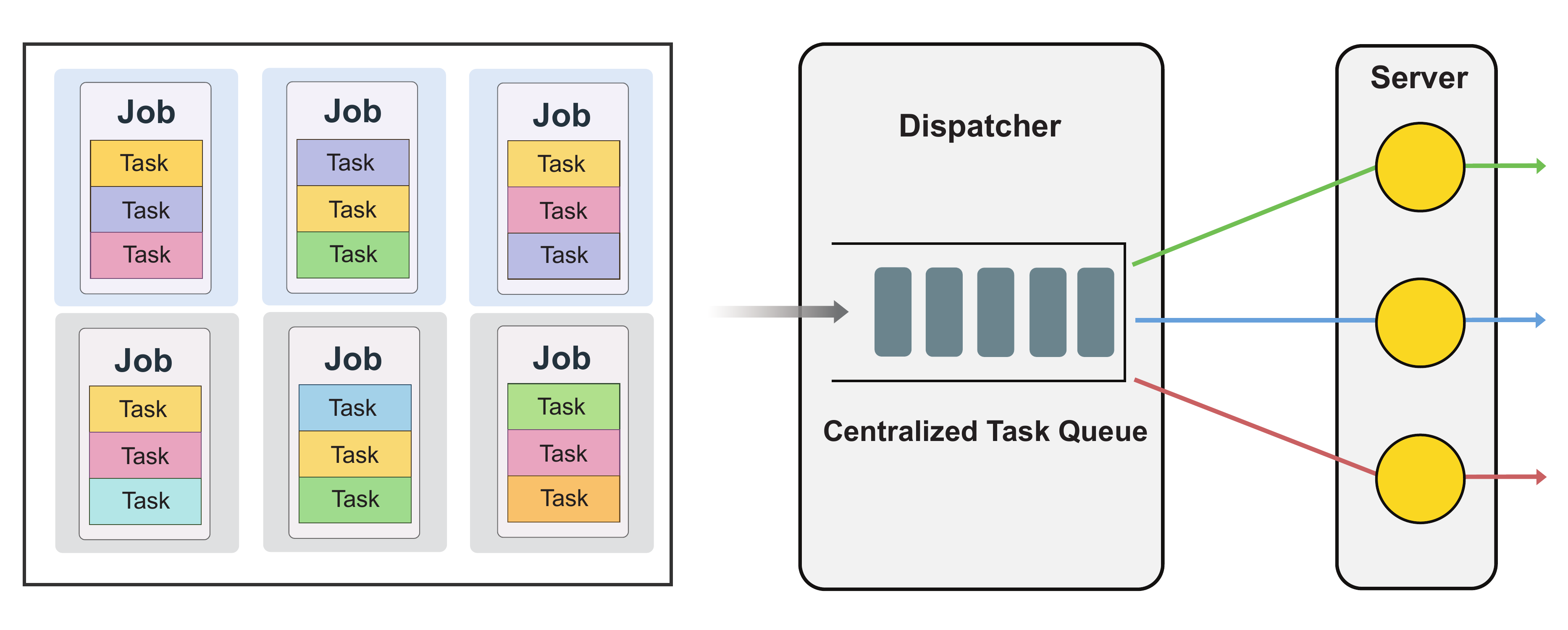}
    \caption{ \color{black}{Illustration of the multi-task multi-processor scheduling model. Jobs consist of multiple (possibly non-preemptive) tasks, which are placed into a single centralized global queue managed by the dispatcher. Servers fetch tasks from the global queue for execution, and a job is considered complete only when all its constituent tasks have finished.}}
    \label{fig:main}
\end{figure*}

However, traditional scheduling problems have evolved significantly from single-task models to complex multitask scenarios. In the contemporary landscape of computational resources, the efficient scheduling of tasks across multiple processors has emerged as a critical challenge, driven by the exponential growth of data and the complexity of applications. To give the simplest example, for the objective of computing matrix vector product, we can divide matrix elements and vector elements into groups of columns and rows respectively, then the tasks correspond to the block-wise multiplication operations. Tasks can also be map, shuffle and reduce procedures in MapReduce framework \cite{wang2016maptask}.

\emph{Multi-processor Multitask Scheduling }(MPMS) \cite{srikant_wang}, as shown in Figure \ref{fig:main}, where a job is considered to be completed only when all the tasks within the job are finished, addresses the allocation of jobs and tasks, each with potentially varying computational demands, to a set of processors in a manner that optimizes overall system performance metrics. Though much progresses have been made in single-task job scheduling, there is a lack of theoretical understanding regarding MPMS. The extension of traditional scheduling principles to multi-task systems, presents significant challenges. The heterogeneity of task durations and the dynamic nature of job arrivals seems further complicate the analysis. Thus a natural question that arises is, 

\begin{center}
\textit{How can we design an efficient scheduling algorithm to minimize the total amount time that the multitask jobs spend in the multi-processor system?}
\end{center}

The significance of MPMS problem is also underscored by its direct impact on deep learning workflow \cite{DBLP:conf/nips/HuangCBFCCLNLWC19, DBLP:conf/sosp/NarayananHPSDGG19, DBLP:conf/osdi/JiangZLYCG20}, several different parallel strategies, such as data parallelism, pipeline parallelism, etc, have been proposed to accelerate the training process. Deep learning models benefit from parallel execution during training and inference phases, necessitating effective job and task scheduling to maximize the utilization of available computational resources. Inefficient scheduling can lead to increased training times, higher operational costs, and suboptimal model performance. Therefore, developing scheduling algorithms that can handle the complexities of modern multiprocessor and multi-task systems is paramount for advancing the field of deep learning.

\subsection{Contributions.}
In this paper, we investigate how to minimize the total response time of multitask jobs in a multi-server system and answer the aforementioned question. Our contributions are summarized as follows.

\begin{itemize}
\item In Section \ref{deterministicalgo}, we propose $\mathsf{NP}$-$\mathsf{SRPT}$ algorithm, for minimizing the total response time. $\mathsf{NP}$-$\mathsf{SRPT}$ algorithm achieves a competitive ratio of $\ln \alpha+\beta+1$, where $\alpha$ is the maximum-to-minimum job workload ratio, $\beta$ represents the ratio between maximum non-preemptive task workload and minimum job workload. It can be shown that no $o(\ln \alpha+\beta)$-competitive algorithm exists when the number of machines is constant. In addition, $O(\ln \alpha+\beta^{1-\varepsilon})$ is the best possible competitive ratio for the class of work-conserving algorithms.

\item Besides the worst case relative ratio above, we further prove our main result, absolute performance guarantees for $\mathsf{NP}$-$\mathsf{SRPT}$ algorithm under certain probabilistic structure on the input instances, in which the relevant and remaining workload bound established for the adversarial inputs contributes significantly to the stochastic analysis. Assuming that jobs arrive according to a Poisson process, \ie, in $\mathsf{M}$/$\mathsf{G}$/$\mathsf{N}$ system, in Section \ref{heavytraffic} we prove that the average response time incurred by $\mathsf{NP}$-$\mathsf{SRPT}$ algorithm is asymptotic optimal when load $\rho \rightarrow 1$, as long as the task size distribution has finite support. As shown in Section \ref{secdis}, the assumption of finite task workload can be removed for exponentially distributed job size, \ie, $\mathsf{M}$/$\mathsf{M}$/$\mathsf{N}$, together with other job size distributions with certain properties on the tail of the distribution, by establishing an upper bound on the busy period in $\mathsf{M}$/$\mathsf{G}$/$\mathsf{N}$.

\item {\color{black} In Section \ref{sec:unknown}, we address the practical challenge of scheduling jobs with unknown sizes. The analysis is extended to show that non-preemptive versions of policies designed for this setting, namely $\mathsf{M\text{-}Gittins}$ and $\mathsf{M\text{-}SERPT}$, are asymptotically optimal and 2-competitive, respectively, in the heavy-traffic limit for a broad class of job size distributions.}

\end{itemize}
In addition, we also further validate the optimality of $\mathsf{NP}$-$\mathsf{SRPT}$ through experiments in Section \ref{exp}, confirming the efficacy of the algorithm. 

\color{black}
\subsection{Analytical challenges in MPMS and our approach}
\label{sec:analytical_challenges}


In many respects, preemptive model are significantly more tractable due to the inherent decomposability of the model. In preemptive-resume disciplines, the highest-priority class can be modeled as a standard $\mathsf{M}/\mathsf{G}/1$ queue, while lower classes can be represented as $\mathsf{M}/\mathsf{G}/1$ queues with busy periods generated by higher-priority classes. This allows the well-developed theory of $\mathsf{M}/\mathsf{G}/1$ queues to be applied iteratively, requiring only aggregate statistics of higher-priority workloads \cite{harchol2013performance}. Such tractability sharply contrasts with non-preemptive scheduling, including the multitask problem considered in this paper, where blocking implies that high-priority jobs may be delayed by the completion of lower-priority non-preemptive tasks already in service, thereby inducing a tight coupling between all priority levels. Thus the performance of a certain class arrival depends not only on the higher-priority classes but also on the detailed instantaneous states of lower-priority jobs, including their remaining service times. This dependence eliminates the possibility of recursive analysis and necessitates richer state descriptions and more complex techniques, such as multidimensional Markov chains \cite{brandwajn2023approximation}.

In multi-processor environments with a finite number of job classes, the analysis quickly becomes intractable due to state-space explosion. To fully describe the system, one must specify the class of the job occupying each server. Even with only five servers and five priority classes, the number of possible states already exceeds one billion~\cite{brandwajn2023approximation}. In multi-task systems, this complexity intensifies, as the system state must additionally capture the internal progress of every active job. The resulting multiplication of the state space by the combinatorial number of internal task configurations leads to an acute curse of dimensionality. Our proposed $\mathsf{NP}$$\text{-}$$\mathsf{SRPT}$ algorithm, with dynamic and infinite priority classes, introduces even more complex scheduling dilemmas.

In addition, the multi-task job structure raises a critical question of priority assignment: \emph{should priority be an attribute of the entire job (job level priority) or of its individual constituent tasks (task level priority)?} A job level policy, where all tasks of a high-priority job inherit its priority, risks resource monopolization if tasks can be parallelized, potentially starving all other jobs. Conversely, a task level policy could cause a critical job to stall if one of its non-essential, low-priority tasks is indefinitely delayed. In this paper, our analysis demonstrates that assigning priority at the job level is, perhaps counter-intuitively, effective in an asymptotic sense.

It is worth mentioning that the seminal work by Leonardi and Raz \cite{leonardi1997approximating} introduced a general technique to transform any preemptive scheduling solution into a non-preemptive one, incurring an additional approximation factor of $O(\sqrt{n/m})$ for total flow time, where $n$ is the number of jobs and $m$ is the number of machines. The approximation ratio is computed by evaluating the increase in total flow time due to the transformation from a preemptive to a non-preemptive solution. However, this transformation relies on complete knowledge of the job set, \ie, total flow time in the preemptive solution, to do job classification, making it inherently offline and limiting its applicability where jobs arrive dynamically. In contrast, our proposed $\mathsf{NP\text{-}SRPT}$ algorithm extends the transformation of preemptive SRPT into a non-preemptive framework in an online-compatible manner.

\color{black}

\subsection{Related Work}
\paragraph{Single and multiple task scheduling} There has been a large literature on single-task job scheduling, with parallel developments taking place in competitive analysis and queuing theory. For example, recently Hong and Scully \cite{hong2024performance} showed Gittins's heavy-traffic optimality in $\mathsf{G}/\mathsf{G}/\mathsf{N}$. However, little is known about multitask scheduling. Scully et. al~\cite{scully2017optimally} presented the first theoretical analysis of \emph{single-processor multitask scheduling} problem, and gave an optimal policy that is easy to compute for batch arrival, together with the assumption that the processing time of tasks satisfies the aged Pareto distributions. Sun et al.~\cite{sun2017near} studied the multitask scheduling problem when all the tasks are of unit size, and proved that among causal and non-preemptive policies, fewest unassigned tasks first (FUT) policy, earliest due date first (EDD) policy, and first come first serve (FCFS) are near delay-optimal in distribution (stochastic ordering) for minimizing the metric of average delay, maximum lateness and maximum delay respectively. Wang et. al \cite{wang2019delay} established results of asymptotic independence of queues for multitask parallel jobs, by developing a new technique named Poisson oversampling. To model the scenario when the scheduler has incomplete information about the job size, Scully et. al~\cite{scully2018optimal} introduced the multistage job model and proposed an optimal scheduling algorithm for multistage job scheduling in $\mathsf{M}$/$\mathsf{G}$/$1$ queue. The closed-form expression of the mean response time is also given for the optimal scheduler. In addition to the aforementioned work, there are also studies that further extend the understanding of scheduling by considering parallelizable jobs represented as Directed Acyclic Graphs (DAGs) \cite{agrawal2016scheduling}. 

\color{black}

\paragraph{Advances in non-preemptive scheduling} Non-preemptive scheduling, characterized by the irrevocable assignment of jobs to uninterrupted execution slots, encounters profound analytical challenges stemming from state-space explosion and the absence of flexible interruption mechanisms, as evidenced in both adversarial and stochastic settings. These difficulties manifest in stringent lower bounds, such as $\Omega(n^{1/2-\varepsilon})$-inapproximability for minimizing total flow-time even on a single machine ~\cite{kellerer1996approximability} and $\Omega(n^{1/3-\varepsilon})$ on parallel machines \cite{leonardi1997approximating}. Several algorithmic results \cite{im2014dynamic, lucarelli2016online, lucarelli2018online, lucarelli2021online} circumvent such barriers through relaxed models like resource augmentation, where algorithms access faster processors than adversaries, and rejection, permitting the discard of a small job fraction.

Non-preemptive multiserver priority queues are considered in queueing theory, which has resisted general exact analysis except in highly constrained regimes, and has led to a long line of special-case solutions and approximations. Exact results are obtained only in very restricted cases, such as identical service-time distributions across classes \cite{Davis1966} or systems with only two priority levels \cite{Williams1980, kim2021non, zuk2024explicit}. Similarly, restrictions to two priority levels have enabled precise computations of performance metrics like waiting times and queue lengths, with recent advancements deriving closed-form expressions via recurrence relations and contour integrals for equal service rates \cite{zuk2024explicit}. When service times differ, results like those of Gail et al. \cite{gail1988analysis} suffer from numerical instability as the number of servers grows. Several approaches exploited structured state-space methods, e.g., quasi-birth-and-death formulations\cite{kao1990computing} or matrix-geometric solutions \cite{reddy1993nonpreemptive}, but these remain feasible only for very small systems. To extend applicability, approximation approaches have emerged for complex arrival patterns, such as quasi-Poisson processes with multiple priority classes \cite{brandwajn2023approximation}.

Applications in modern computing further illustrate progress, practical extensions in cloud environments have adapted non-preemptive models for virtual machine scheduling, leveraging knapsack-based algorithms to balance throughput and delay \cite{psychas2017non}. Investigations into capacity rationing in multiserver nonpreemptive priority queues, where low-priority access is restricted unless idle servers exceed a threshold, provide asymptotic insights for large-scale systems with closely matched supply and demand \cite{baron2025capacity}.

All these contributions advance performance analysis, but also underscore the need for methods in non-preemptive setting to handle diverse service distributions and priority structures.

\color{black}

\paragraph{Performance and optimality of SRPT and its variants} While SRPT is optimal at minimizing average response time in single-server systems \cite{smith1978new}, its performance is suboptimal in multi-server environments. However, SRPT remains a highly regarded method in the context of multi-server systems. It has shown that the competitive ratio of SRPT is in logarithmic order and is the best possible \cite{leonardi1997approximating}. With the stochastic arrivals in $\mathsf{M}$/$\mathsf{G}$/$\mathsf{N}$ systems, SRPT is shown to be optimal in heavy traffic scenarios \cite{grosof2018srpt}.
Another notable contribution by \cite{im2016competitively} introduces the Intermediate-SRPT algorithm, which addresses jobs with intermediate parallelizability and establishes an optimal $O(\log P)$-competitive ratio concerning average flow time, where $P$ represents the ratio of maximum to minimum job sizes.

\subsection{Paper organization} The remainder of this paper is organized as following. We introduce the problem definition, notations and necessary background in Section \ref{modelpre}. In Section \ref{deterministicalgo} we formally present $\mathsf{NP}$-$\mathsf{SRPT}$ algorithm, together with the analysis of its competitive ratio and lower bounds. Section \ref{heavytraffic} is devoted to the proof of the asymptotic optimality of $\mathsf{NP}$-$\mathsf{SRPT}$ in heavy traffic regime, and the optimality is extended to infinite task size in Section \ref{secdis}. In Section \ref{sec:unknown}, the analysis is extended to the setting
of unknown job sizes. We conduct the experimental validation in Section \ref{exp}, and conclude our work in Section \ref{seccon}.

\section{Model and preliminaries}\label{modelpre}
\paragraph{Deterministic Model}
We are given a set $\mathcal{J}=\{\mathsf{J}_{1}, \mathsf{J}_{2}, \ldots, \mathsf{J}_{n}\}$ of $n$ jobs arriving online over time, together with a set of $N$ identical machines. Job $i$ consists of $n_{i}$ tasks and its workload $p_{i}$ is equal to the total summation of the processing time of tasks, \ie, $p_{i}=\sum_{\ell \in [n_{i}]}{p_{i, \ell}}$, where $p_{i, \ell}$ represents the processing time of the $\ell$-th task of job $i$. Tasks can be either preemptive or non-preemptive. A task is non-preemptive if it is not allowed to interrupt the task once it starts service, \ie, the task is run to completion. All the information of job $i$ is unknown to the algorithm until its releasing date $r_{i}$. Under any given scheduling algorithm, the completion time of job $j$ under the algorithm, denoted by $C_{j}$, is equal to the maximum completion time of individual tasks within the job. Formally, let $C^{(\ell)}_{j}$ be the completion time of task $\ell$ in job $j$, then $C_{j}=\max_{\ell\in [n_{i}]}{C^{(\ell)}_{j}}$. The response time of job $j$ is defined as $F_{j}=C_{j}-r_{j}$, our objective is to minimize the total response time $\sum_{j\in [n]}{F_{j}}$. 

Throughout the paper we use $\alpha=\max_{i\in [n]}{p_{i}}/\min_{i\in [n]}{p_{i}}$ to denote the ratio of the maximum to the minimum job workload. Let $\eta=\max\{p_{i, \ell}| \mbox{ task } \ell \mbox{ of job } i \mbox{ is non-preemptive}\}$ be the maximum processing time of a non-preemptive task, $\beta=\eta/\min_{i\in [n]}{p_{i}}$ be the ratio between $\eta$ and minimum job workload. In some sense, parameters $\beta$ and $\eta$ represent the degree of non-preemptivity and exhibits a trade-off between the preemptive and non-preemptive setting. The problem approaches the preemptive case when $\eta$ is small, and degenerates to the non-preemptive case if all the jobs are consisted of a single non-preemptive task, in which $\eta$ reaches the maximum value of $\max_{i\in [n]}{p_{i}}$.

\begin{definition}[Work-conserving scheduling algorithm] A scheduling algorithm $\pi$ is called work-conserving if it fills maximally when there exists at least one feasible job or task awaiting the execution in the system. Here a job or task is called feasible, if it satisfies all the given constraints of the system (e.g, preemptive and non-preemptive constraint, precedence constraint, etc).
\end{definition}

\begin{definition}[Competitive ratio] The competitive ratio of online algorithm $\mathcal{A}$ refers to the worst ratio of the cost incurred by $\mathcal{A}$ and that of optimal offline algorithm $\mathcal{A}^{*}$ over all input instances $\omega$ in $\Omega$, \ie,
\begin{align*}
\mathcal{CR}_{\mathcal{A}}=\max_{\omega \in \Omega}	\frac{\mathrm{Cost}_{\mathcal{A}}(\omega)}{\mathrm{Cost}_{\mathcal{A^{*}}}(\omega)}.
\end{align*}
In the multiple-processor multitask scheduling problem, the cost is the total response time under instance $\omega=\{(r_{i}, \{p_{i, \ell}\}_{\ell \in [n_{i}]})\}_{i\in [n]}$.	
\end{definition}

\paragraph{Stochastic Model}
In the stochastic setting, we assume that jobs arrive into the system according to a Poisson process with rate $\lambda$. Job processing times are i.i.d distributed with probability density function $f(\cdot)$. Formally,  we consider a sequence of $\mathsf{M}$/$\mathsf{G}$/$\mathsf{N}$ queues indexed by $k$, the traffic intensity of the $k$-th system is equal to $\rho^{(k)}=\lambda^{(k)}\cdot \mathbbm{E}[p^{(k)}_{i}]$, where $\lambda^{(k)}$ denotes the arrival rate of the $k$-th Poisson arrival process, job workload distribution has a density function of $f^{(k)}(\cdot)$. Stability of the queuing systems requires that $\rho^{(k)}<1$ for $\forall k$. As standard in the literature, we assume that $\rho^{(k)}\rightarrow1$ when $k\rightarrow \infty$. In this paper, we further assume that the probability density function $f^{(k)}(\cdot)$ is continuous. For notational convenience, we will suppress index $n$ whenever it is clear from the context. 

The stochastic analysis in this paper relies heavily on the concept of busy period in $\mathsf{M}$/$\mathsf{G}$/$1$, which is defined as following.

\begin{definition}[Busy Period in $\mathsf{M}$/$\mathsf{G}$/$1$~\cite{harchol2013performance}]
Busy period in $\mathsf{M}$/$\mathsf{G}$/$1$ is defined to be the longest time interval in which no machines are idle. 
\end{definition}
\noindent We use $\mathsf{B}(w)$ to denote the expected length of a busy period in $\mathsf{M}$/$\mathsf{G}$/$1$, started by a workload of $w$.
It can be seen that $\mathsf{B}(\cdot)$ is an additive function~\cite{harchol2013performance}, \ie, $\mathsf{B}(w_{1}+w_{2})=\mathsf{B}(w_{1})+\mathsf{B}(w_{2})$ for any independent initial workload $w_{1}, w_{2}$, 
since a busy period in $\mathsf{M}$/$\mathsf{G}$/$1$ with initial workload of $w_{1}+w_{2}$ can be regarded as a busy period started by initial workload $w_{2}$, following a busy period started by initial workload $w_{1}$. Moreover, for $\mathsf{M}$/$\mathsf{G}$/$1$ queue, the length of a busy period with initial workload of $w$ and load $\rho$ is shown to be equal to $\mathsf{B}(w)=\mathbbm{E}[w]/(1-\rho)$ \cite{harchol2013performance}.

\paragraph{SOAP policy} SOAP (Schedule Ordered by Age-based Priority) \cite{DBLP:conf/sigmetrics/ScullyHS18} is a comprehensive framework that enables the analysis of a vast array of scheduling policies, encompassing both well-established methods like FCFS and SRPT, as well as innovative variants with dynamic rank functions, which is formally defined as following.

\begin{definition}[SOAP Policy \cite{DBLP:conf/sigmetrics/ScullyHS18}]
A SOAP scheduling policy is a preemptive priority policy that utilizes a job's type (descriptor $d$) and its age $a$ (the amount of time it has been served in the system) to determine its priority. Policy $\pi$ employs a rank function $r^{\pi}(d, a)$ that assigns a numerical value (rank) to each job based on its type and age. The job with the lowest rank is always selected for service. 
\end{definition}
A SOAP policy is monotonic if its rank function is monotonic in age $a$. In this paper, a tie-breaking rule, First-Come First-Served (FCFS) is applied to determine the order of service in the event of a tie. {\color{black}{We use \textbf{$p$} to denote the specific descriptor $d$ for our size-based ploicy.}}

\section{$\mathsf{NP}$-$\mathsf{SRPT}$ Algorithm and Competitive Ratio Analysis}\label{deterministicalgo}

In the realm of scheduling, non-preemptive policies is required when tasks must be completed without interruption once started. To address this constraint, we introduce \emph{Non-Preemptive-$\pi$} ($\mathsf{NP}$-$\pi$) policy, which slightly modifies the rank function of any given SOAP policy $\pi$ to account for non-preemptive tasks.

\begin{definition}[$\mathsf{NP}$-$\pi$ policy]
For any SOAP policy $\pi$, the rank function of $\mathsf{NP}$-$\pi$ policy is,
\begin{align*}
&r^{\mathsf{NP}\text{-}{\pi}}(d, a) \\
=& r^{\pi}(d, a)\cdot \mathbbm{1}(\nexists \text{ started but unfinished non-preemptive task}).
\end{align*}
\end{definition}
This rank function adjusts the priority of a job based on whether the current non-preemptive task has been completed. If a non-preemptive task is still ongoing, the rank remains unaltered to be zero; otherwise, the rank is adjusted according to the original SOAP policy $\pi$. 

To operationalize the $\mathsf{NP}$-$\pi$ policy, we propose the $\mathsf{NP}$-$\mathsf{SRPT}$  algorithm, an adaptation of SRPT algorithm that respects non-preemptive constraints. 
The rank function of $\mathsf{NP}$-$\mathsf{SRPT}$ policy is as following:
\begin{align*}
&r^{\mathsf{NP\text{-}SRPT}}(p, a) \\
= &(p-a) \cdot \mathbbm{1}(\nexists \text{ started but unfinished non-preemptive task}),
\end{align*}
where $p$ and $a$ denotes the total workload and age of the job respectively. At each time slot $t$, jobs with non-preemptive task are kept processing on the machines, while the remaining machines are used to process jobs with smallest remaining workload. The main idea of $\mathsf{NP}$-$\mathsf{SRPT}$  is similar to SRPT, \ie, we utilize as many resources as possible on the job with smallest remaining workload, to reduce the number of alive jobs in a greedy manner, while satisfying the non-preemptive constraint.

\subsection{A general upper bound on relevant work for $\mathsf{NP}$-$\pi$ policy}
In the following sections, we will delve into the details of the analysis of $\mathsf{NP}$-$\mathsf{SRPT}$ algorithm, elucidating its behavior and performance guarantees. Central to the analysis are Lemma \ref{generallemma} and Lemma \ref{remaining_workload}, which compare the algorithm’s performance under non-preemptive constraints with optimal preemptive algorithm. Before we present these lemmas, however, it is necessary to introduce and define three important concepts, 
\emph{relevant $y$-work, system relevant $y$-work} and \emph{old job age cut off}.

\begin{definition}[Relevant $y$-work]
For a job $\mathsf{J}$ with rank no more than $y$ and SOAP policy $\pi$, its relevant $y$-work refers to the quantity of service required by a job until it enters a state with a rank of at least $y$ or completes, \ie, 
\color{black}
\begin{align*}
&z^{\pi}_{\mathsf{J}}(y, t)\\
= &[\sup \{a\in [0, p_{\mathsf{J}}]|r^{\pi}(p_{\mathsf{J}}, a) \leq y\} - a_{\mathsf{J}}(t)]\cdot \mathbbm{1}({r^{\pi}_{\mathsf{J}}(p_{\mathsf{J}}, a_{\mathsf{J}}(t))\leq y}),
\end{align*}
where $a_{\mathsf{J}}(t)$ denotes the age of job $\mathsf{J}$ at time $t$.
\color{black}
\end{definition}
We can see that, for monotone decreasing SOAP policy $\pi$, such as SRPT, the relevant $y$-work of a job is the amount of service it needs to complete the job.

System relevant $y$-work extends the concept of relevant $y$-work to the entire system, which is defined as the total of the remaining $y$-work across all jobs currently in the system:
\begin{align*}
\mathsf{Rel}^{\pi}_{\leq y}(t) = \sum_{\mathsf{J}}{z^{\pi}_{\mathsf{J}}(y, t)}.
\end{align*}
The concept of old job age cut off is a threshold used to determine the maximum age of a job that can be considered for scheduling. Jobs that exceed this age cut off are deemed irrelevant and may be subject to lower priorities. 
\begin{definition}[Old job age cut off]
For any monotone non-decreasing policy $\pi$, we let
\begin{align*}
z^{\pi}_{\mathsf{J}}(y) = \sup \{a\geq 0|r^{\pi}(p_{\mathsf{J}}, a) \leq y\},
\end{align*}
and for monotone decreasing policy $\pi$, we let 
\begin{align*}
z^{\pi}_{\mathsf{J}}(y) = \sup \{a\geq 0|r^{\pi}(p_{\mathsf{J}}, p_{\mathsf{J}}-a) \leq y\}.
\end{align*}
\end{definition}
Here we extend the definition of old job age cutoff in \cite{scully2021optimal} for allowing monotone decreasing policy. It can be seen that 
\begin{align*}
z^{\pi}_{\mathsf{J}}(y, t) \leq z^{\pi}_{\mathsf{J}}(y)
\end{align*}
holds for all monotone policy $\pi$.
\begin{lemma}\label{generallemma}
For any monotone SOAP policy $\pi$ without considering the non-preemptive constraint, the system $y$-work under policy $\mathsf{NP}$-$\pi$ satisfies that
\begin{align*}
\mathsf{Rel}^{\mathsf{NP}\text{-}\pi}_{\leq y}(t)-\mathsf{Rel}^{\pi^{*}_{1,N}}_{\leq y}(t)\leq (N-1)\cdot (\max_{\mathsf{J}}z^{\pi}_{\mathsf{J}}({y})+\eta), \forall r,t \geq 0,
\end{align*}
where $\pi^{*}_{1,N}$ represents the optimal algorithm in single server with speed $N$. 
\end{lemma}

\begin{proof}
We analyze the following two cases: 
\begin{itemize}
\item \emph{few-jobs time intervals}, in which there are $N-1$ or fewer jobs of rank no more than $y$. Consider one of such job $\mathsf{J}$, if one of the non-preemptive tasks of this job is under processing, then 
\begin{align*}
z^{\mathsf{NP}\text{-}\pi}_{\mathsf{J}}(y,t)\leq \max_{\mathsf{J}}z^{\pi}_{\mathsf{J}}({y}) + \eta,
\end{align*}
if the rank of this job does not exceed $y$ after completing the non-preemptive task,
otherwise we have
\begin{align*}
z^{\mathsf{NP}\text{-}\pi}_{\mathsf{J}}(y,t)\leq  \eta.
\end{align*}
If this job does not have a non-preemptive task under processing, then we have $z^{\mathsf{NP}\text{-}\pi}_{\mathsf{J}}(y,t)\leq \max_{\mathsf{J}}z^{\pi}_{\mathsf{J}}({y})$. 

Hence for time slots $t$ in few jobs intervals,
\begin{align*}
&\mathsf{Rel}^{\mathsf{NP}\text{-}\pi}_{\leq y}(t)-\mathsf{Rel}^{\pi^{*}_{1,N}}_{\leq y}(t)\\
\leq &\mathsf{Rel}^{\mathsf{NP}\text{-}\pi}_{\leq y}(t)\leq (N-1)\cdot (\max_{\mathsf{J}}z^{\pi}_{\mathsf{J}}({y})+\eta).
\end{align*}
\item \emph{many-jobs time intervals}, in which all $N$ servers are serving jobs of rank no more than $y$. Note that all $N$ servers are working towards reducing the value of $\mathsf{Rel}^{\mathsf{NP}\text{-}\pi}_{\leq y}(t)$, while the two systems experience the same arrival sequence. Therefore, in the many interval period, the rate of decrease of $\mathsf{Rel}^{\mathsf{NP}\text{-}\pi}_{\leq y}(t)$ in the $N$-server system is no less than the rate of decrease of $\mathsf{Rel}^{\pi^{*}_{1,N}}_{\leq y}(t)$ in single server system. Consequently, we have the following:
\begin{align*}
&\mathsf{Rel}^{\mathsf{NP}\text{-}\pi}_{\leq y}(t)-\mathsf{Rel}^{\pi^{*}_{1,N}}_{\leq y}(t)\\
\leq &\mathsf{Rel}^{\mathsf{NP}\text{-}\pi}_{\leq y}(t^{\dag})-\mathsf{Rel}^{\pi^{*}_{1,N}}_{\leq y}(t^{\dag})\\
\leq &(N-1)\cdot (\max_{\mathsf{J}}z^{\pi}_{\mathsf{J}}({y})+\eta),
\end{align*}
where $t^{\dag}$ represents the moment just before the start of the many job interval in which $t$ is located, specifically the last moment of the few job interval preceding $t$.
\end{itemize}
\end{proof}

\subsection{Performance Analysis of $\mathsf{NP}$-$\mathsf{SRPT}$}
Having introduced the bound on the difference of relevant work under a general $\mathsf{NP}$-$\pi$ policy and optimal, we now proceed to discuss the bound on the remaining workload.

\begin{lemma}\label{remaining_workload}
For any $y\geq 0$, let $W^{\pi}_{\leq y}(t)$ represent the total remaining workload at time $t$ of jobs under policy $\pi$ whose remaining workload does not exceed $y$, then
\begin{align*}
\mathsf{W}^{ \mathsf{NP\text{-}SRPT}}_{\leq y}(t) - \mathsf{W}^{\pi^{*}_{1,N}}_{\leq y}(t) \leq N\cdot (y+\eta).
\end{align*}
\end{lemma}

\begin{proof}
As SRPT is a monotone decreasing SOAP policy, we have 
\begin{align*}
z^{\mathrm{SRPT}}_{\mathsf{J}}(y) = y, \forall \mathsf{J}.
\end{align*}
Hence, according to Lemma \ref{generallemma}, we further have:
\begin{align}\label{rel-np-srpt}
\mathsf{Rel}^{ \mathsf{NP\text{-}SRPT}}_{\leq y}(t)-\mathsf{Rel}^{\pi^{*}_{1,N}}_{\leq y}(t)\leq (N-1)\cdot (y+\eta).
\end{align}
Note that the relevant work can be decomposed as the sum of relevant workload of jobs with rank in $(0, y]$ (denoted as $\mathsf{W}_{\leq y}^{\pi}(t)$) and the relevant workload of jobs with rank 0 (denoted as $\mathsf{V}_{\leq y}^{\pi}(t) $). In the context of $\mathsf{NP}$-$\mathsf{SRPT}$ , it is observed that when a job's rank is greater than zero, its relevant y-work is equivalent to its remaining workload. Therefore, this decomposition can be expressed as follows:
\color{black}
\begin{align}
\mathsf{Rel}_{\leq y}^{\mathsf{NP}\text{-}\mathsf{SRPT}}(t) = \mathsf{W}_{\leq y}^{\mathsf{NP}\text{-}\mathsf{SRPT}}(t) + \mathsf{V}_{\leq y}^{\mathsf{NP}\text{-}\mathsf{SRPT}}(t). 
\end{align}
\color{black}
Consequently 
\begin{align*}
&\mathsf{W}^{ \mathsf{NP\text{-}SRPT}}_{\leq y}(t) - \mathsf{W}^{\pi^{*}_{1,N}}_{\leq y}(t)\\
= &\mathsf{Rel}^{ \mathsf{NP\text{-}SRPT}}_{\leq y}(t) - \mathsf{Rel}^{\pi^{*}_{1,N}}_{\leq y}(t) + \mathsf{V}^{\pi^{*}_{1,N}}_{\leq y}(t) - \mathsf{V}^{ \mathsf{NP\text{-}SRPT}}_{\leq y}(t)\\
\leq & \mathsf{Rel}^{ \mathsf{NP\text{-}SRPT}}_{\leq y}(t) - \mathsf{Rel}^{\pi^{*}_{1,N}}_{\leq y}(t) + \mathsf{V}^{\pi^{*}_{1,N}}_{\leq y}(t).
\end{align*}
Since under policy $\pi^{*}_{1,N}$, there is at most one job with a rank of $0$, and its relevant $y$-work does not exceed $\eta + y$. Combined with (\ref{rel-np-srpt}), the proof is complete.
\end{proof}

Building on the previous discussions, we now turn our attention to proving the performance of the $\mathsf{NP}$-$\mathsf{SRPT}$  Algorithm, specifically its competitive ratio. Utilizing the established conclusions, we will demonstrate how these bounds contribute to the algorithm's effectiveness. 

\begin{lemma}\label{num_of_job_bound}
Let $n_{\pi}(t)$ represents the number of alive jobs at time $t$ under policy $\pi$, then 
\begin{align*}
n_{\mathsf{NP\text{-}SRPT}}(t) \leq n_{\pi^{*}_{1,N}}(t)+N\cdot (\ln \alpha+\beta+1).
\end{align*}
\end{lemma}

\begin{proof}
In Appendix \ref{appendixcomproof} we present a self-contained proof, which shows that $n_{\mathsf{NP\text{-}SRPT}}(t) \leq n_{\pi^{*}_{1,N}}(t)+N\cdot (\ln \alpha+\beta+1)$. The main idea is to divide the jobs into different classes and compare the remaining number of jobs under $\mathsf{NP}$-$\mathsf{SRPT}$  with that under optimal algorithm $\pi^{*}_{1,N}$. For any algorithm $\pi$, at time slot $t$, we divide the unfinished jobs into $\Theta(\log \alpha)$ classes $\{\mathcal{C}_{k}(\pi, t)\}_{k\in [\log \alpha +1]}$, based on their remaining workload. Jobs with remaining workload that is no more than $2^{k}$ and larger than $2^{k-1}$ are assigned to the $k$-th class. Formally, 
\begin{align*}
\mathcal{C}_{k}(\pi,t)=\Big\{i\in [n]\;\Big|\;W_i(\pi,t)\in (2^{k-1},2^{k}]\Big\},
\end{align*}
where $W_{i}(\pi,t)$ represents the unfinished workload of job $i$ at time $t$ under policy $\pi$. We finish the proof via similar approaches to \cite{leonardi2007approximating}, which primarily hinges on the fact that jobs within the same class $\mathcal{C}_{k}(\pi,t)$ have sizes that are relatively close, differing by at most a constant factor of $2$. This allows us to bound the number of unfinished jobs by considering their remaining workloads and using Lemma \ref{remaining_workload}. 

Here we present an alternative proof that results in a tighter bound with smaller coefficients, using the \emph{WINE Identity} introduced in \cite{scully2020gittins, scully2023new}. 

\begin{lemma}[WINE Identity \cite{scully2020gittins, scully2023new}] For any scheduling policy $\pi$, 
\begin{align*}
n_{\pi} &= \int_{0}^{\infty}{\frac{\mathbbm{E}[\mathsf{Rel}^{\pi}_{\leq r}|\text{State of Jobs } \{\mathsf{J}\}^{n}_{i=1}]}{r^{2}}} \mathrm{d}r,\\
\mathbbm{E}[n_{\pi}] &= \int_{0}^{\infty}{\frac{\mathbbm{E}[\mathsf{Rel}^{\pi}_{\leq r}]}{r^{2}}} \mathrm{d}r
\end{align*}
holds for $\mathsf{G}/\mathsf{G}/\mathsf{N}$ system.
\end{lemma}
\begin{align*}
n_{ \mathsf{NP\text{-}SRPT}}(t) &\leq  \int_{0}^{\infty}{\frac{\mathbbm{E}[\mathsf{Rel}^{\pi^{*}}_{\leq r}]}{r^{2}}} \mathrm{d}r +   \int_{0}^{\infty} { \frac{N(r+\eta)}{r^{2}}}\mathrm{d}r\\
&\leq n_{\pi^{*}_{1,N}}(t) + N\int_{p_{\min}}^{p_{\max}} \Big(\frac{1}{r}+\frac{\eta}{r^{2}}\Big)\mathrm{d}r \\
&=n_{\pi^{*}_{1,N}}(t) + N\cdot \Big(\ln \alpha + \frac{\eta}{p_{\min}} -\frac{\eta}{p_{\max}}\Big)\\
& \leq n_{\pi^{*}_{1,N}}(t) + N\cdot \Big(\ln \alpha + \beta \Big).
\end{align*}
\end{proof}
Our main result is stated in the following theorem.
\begin{theorem}\label{algotheorem}
 $\mathsf{NP}$-$\mathsf{SRPT}$ Algorithm achieves a competitive ratio that is no more than 
\begin{align*}
\mathcal{CR}_{ \mathsf{NP\text{-}SRPT}}\leq \ln \alpha+\beta+1.
\end{align*}
\end{theorem}
\begin{proof}
Again we divide the time slots into few-job intervals and many-job intervals, the competitive ratio of $\mathsf{NP}$-$\mathsf{SRPT}$  satisfies that
\begin{align*}
&\mathcal{CR}_{\mathsf{NP\text{-}SRPT}}\\
=&\frac{\int_{t: n_{ \mathsf{NP\text{-}SRPT}}(t)<N}{n_{ \mathsf{NP\text{-}SRPT}}(t)}}{F^{\pi^{*}_{1,N}}}\\
&+\frac{\int_{t: n_{ \mathsf{NP\text{-}SRPT}}(t)\geq N}{n_{ \mathsf{NP\text{-}SRPT}}(t)}}{F^{\pi^{*}_{1,N}}}\\
\leq & \frac{\int_{t: n_{ \mathsf{NP\text{-}SRPT}}(t)\geq N}{n_{\pi^{*}}(t)}}{F^{\pi^{*}_{1,N}}} +\Big(\ln \alpha+\beta\Big)\cdot \frac{\int_{t: n_{ \mathsf{NP\text{-}SRPT}}(t)\geq N}{N}}{F^{\pi^{*}_{1,N}}}\\
&+\frac{\int_{t: n_{ \mathsf{NP\text{-}SRPT}}(t)<N}{n_{ \mathsf{NP\text{-}SRPT}}(t)}}{F^{\pi^{*}_{1,N}}}\\
\leq &\ln \alpha+\beta+1,
\end{align*}
where the first inequality follows from Lemma \ref{num_of_job_bound}, the second inequality is due to the fact that
\begin{align*}
\int_{t: n_{ \mathsf{NP\text{-}SRPT}}(t)\geq N}{N}+\int_{t: n_{ \mathsf{NP\text{-}SRPT}}(t)<N}{n_{ \mathsf{NP\text{-}SRPT}}(t)}
\end{align*}
and 
\begin{align*}
\int_{t: n_{ \mathsf{NP\text{-}SRPT}}(t)\geq N}{n_{\pi^{*}_{1,N}}(t)}
\end{align*}
are two lower bounds of minimum total response time. As $F^{\pi^{*}_{1,N}}\leq F^{\pi^{*}}$, the proof is complete.
\end{proof}

\subsection{Competitive ratio lower bound}
The following lower bounds mainly follow from the observation that, multiple-processor multitask scheduling problem generalizes the single-task job scheduling problem in both preemptive and non-preemptive settings.
\begin{fact}\label{lowerboundlemma1}
For multiple-processor multitask scheduling problem with constant number of machines, there exists no algorithm that achieves a competitive ratio of $o(\ln \alpha+\beta)$.	
\end{fact}
\begin{proof}
When $\eta=0$, the problem degenerates to preemptive setting and no algorithm can achieve a competitive ratio of $o(\ln \alpha)$ \cite{leonardi2007approximating}. When $\eta=p_{\max}$, the problem degenerates to the non-preemptive setting and $O(\beta)$ is the best possible competitive ratio if the number of machines is constant \cite{bunde2002approximating}. The proof is complete.
\end{proof}

\begin{fact}\label{lowerboundlemma2}
For multiple-processor multitask scheduling problem, the competitive ratio of any work-conserving algorithms have an competitive ratio of $\Omega(\ln \alpha+ \beta^{1-\varepsilon})$ for $\forall \varepsilon>0$.	
\end{fact}
\begin{proof}
The reasoning is similar as the proof of Proposition~\ref{lowerboundlemma1}, since work-conserving algorithms cannot achieve a competitive ratio of $o(\beta^{1-\varepsilon})$ in the non-preemptive single-task job scheduling \cite{bunde2002approximating}.
\end{proof}

\section{Asymptotic Optimality of $\mathsf{NP}$-$\mathsf{SRPT}$  with Poisson Arrival}\label{heavytraffic}
In this section we show that under mild probabilistic assumptions, $\mathsf{NP}$-$\mathsf{SRPT}$  is asymptotic optimal for minimizing the total response time in the heavy traffic regime. The result is formally stated as following.

\begin{theorem}\label{heavytrafficthm}
Let $F^{ \mathsf{NP\text{-}SRPT}}_{\rho}$ and $F^{\pi^{*}}_{\rho}$ be the response time incurred by $\mathsf{NP}$-$\mathsf{SRPT}$  and optimal algorithm respectively, when the traffic intensity is equal to $\rho$. In an $\mathsf{M}$/$\mathsf{G}$/$\mathsf{N}$ with finite task size and job size distribution satisfying $\mathbbm{E}[p^{2}(\log p)^{+}]<\infty$, $\mathsf{NP}$-$\mathsf{SRPT}$  is heavy traffic optimal, \ie,
\begin{align}\label{taskthreshold}
\lim\nolimits_{\rho \rightarrow 1}\frac{\mathbbm{E}[F^{\mathsf{NP\text{-}SRPT}}_{\rho}]}{\mathbbm{E}[F^{\pi^{*}}_{\rho}]}=1.
\end{align}

\end{theorem}

The probabilistic assumptions here are with respect to the distribution of job size, \ie, the total workload of tasks. For the processing time of a single task, the only assumption we have is the finiteness of task workload. It can be seen that the optimality result in~\cite{grosof2018srpt} corresponds to a special case of Theorem \ref{heavytrafficthm}.

\color{black}

\color{black}

\subsection{Average response time bound}\label{uplowbound}
Our important step is to derive the following analytical upper bound on $\mathbbm{E}[F_{\rho}^{ \mathsf{NP\text{-}SRPT}}]$. 
\begin{theorem}\label{lemma1}
The average response time under $\mathsf{NP}$-$\mathsf{SRPT}$  satisfies that
\begin{align}\label{unrefine1}
\mathbbm{E}{[F_{\rho}^{ \mathsf{NP\text{-}SRPT}}]}\leq \mathbbm{E}{[F_{\rho}^{\mathsf{SRPT}_{1,N}}]}+O\Big( \log \frac{1}{1-\rho}\Big)\cdot \mathbbm{E}[\eta].
\end{align}
\end{theorem}

\begin{proof}

{\color{black}{In this proof, we focus on the careful application of the technique in \cite{grosof2018srpt} to the complex $\mathsf{NP\text{-}SRPT}$ setting, correctly accounting for the additional delays caused by non-preemptivity and the multi-processor environment.

\noindent\emph{Proof Sketch:} Similar as~\cite{grosof2018srpt}, the core idea is that the tagged job's response time is determined by the time required for the $N$ servers to process a total volume of work, which includes the job's own workload plus all other work that takes priority. By establishing upper bounds on each component and leveraging $\mathsf{M}/\mathsf{G}/1$ busy period properties, we derive a bound on the response time. Finally, we apply the PASTA property to relate this to the average response time of an ideal $\mathsf{SRPT}$ system.}

\paragraph{Workload Decomposition for a Tagged Job and Bounding the Workload Components}

Consider a tagged job with workload $x$, arriving time $r_{x}$ and completion time $C^{\mathsf{NP\text{-}SRPT}}_{x}$. Its response time is $\mathsf{T}^{\mathsf{NP\text{-}SRPT}}_{x} = C^{\mathsf{NP\text{-}SRPT}}_{x} - r_{x}$. This duration is the time for the \(N\) servers to clear a specific volume of work. we categorize this work into four components processed during during $[r_{x}, C^{\mathsf{NP\text{-}SRPT}}_{x}]$:

}

\begin{enumerate}
\item \emph{\color{black}{Wasted service capacity.}} The system may be processing jobs with rank larger than $x$, or some machines are idle, while the tagged job is in service, because the number of jobs alive is smaller than $N$. We use $\mathsf{W}_{\mathrm{waste}}(r_{x})$ to represent the amount of such resources, then
\begin{align}\label{wastebound}
\mathsf{W}_{\mathrm{waste}}(r_{x})\leq (N-1)\cdot x,
\end{align}
which is indeed the same as Lemma $5.1$ in~\cite{grosof2018srpt}. The reason is straightforward---as the tagged job must be in service, hence the number of such time slots should not exceed $x$, and thus (\ref{wastebound}) holds.
\item \emph{\color{black}{Relevant work at arrivals.}} The servers may be dealing with system relevant $x$-work, which is of a higher priority than the tagged job. According to Lemma \ref{generallemma} and Lemma \ref{remaining_workload}, 
\begin{align}
&\mathsf{Rel}^{ \mathsf{NP\text{-}SRPT}}_{\leq x}(r_{x})\\
\leq &\mathsf{Rel}^{\mathsf{SRPT}_{1,N}}_{\leq x}(r_{x}) + (N-1)\cdot (x+\eta)\notag\\
= &\mathsf{W}^{\mathsf{SRPT}_{1,N}}_{\leq x}(r_{x}) + \mathsf{V}^{\mathsf{SRPT}_{1,N}}_{\leq x}(r_{x}) + (N-1)\cdot (x+\eta) \notag\\
\leq & \mathsf{W}^{\mathsf{SRPT}_{1,N}}_{\leq x}(r_{x}) + N \cdot (x+\eta).\label{rele_bound}
\end{align}

\item \emph{\color{black}{New relevant work}}. A newly arriving job is admitted to the system during $[r_{x},C_{x}]$, only if its size is no more than $x$. Hence we only consider relevant load $\rho(x)=\lambda \cdot \int_{0}^{x}{tf(t)dt}$.

\item \emph{\color{black}{Tagged job's workload}}. The amount of resources is equal to $x$, the size of the tagged job.
\end{enumerate}

\color{black}
\paragraph{{Constructing the Busy Period Bound}}
The response time $\mathsf{T}^{ \mathsf{NP\text{-}SRPT}}_{x}$ of the tagged job $x$, is stochastically bounded by the length of a busy period in a single server system with speed $N$, driven by the relevant arrival load \(\rho(x)\), starting with the sum of the tagged job's work, wasted capacity, and pre-existing higher-priority work, \ie, 
\color{black}
\begin{align*}
\mathsf{W}_{\mathrm{waste}}(r_{x})+\mathsf{Rel}^{ \mathsf{NP\text{-}SRPT}}_{\leq x}(r_{x})+x.
\end{align*}
Formally, we have
\begin{align*}
&\mathsf{T}^{ \mathsf{NP\text{-}SRPT}}_{x}\\
\leq_{st}&\mathsf{B}^{(\rho({x}))}\Big(\mathsf{W}_{\mathrm{waste}}(r_{x})+\mathsf{Rel}^{ \mathsf{NP\text{-}SRPT}}_{\leq x}(r_{x})+x\Big)\\
\overset{(a)}{\leq} &\mathsf{B}^{(\rho({x}))}\Big(2Nx+N\eta+\mathsf{W}^{\mathsf{SRPT}_{1,N}}_{\leq x}(r_{x})\Big)\\
\overset{(b)}{=} &\underbrace{\mathsf{B}^{(\rho({x}))}\Big(2Nx+N\eta\Big)}_{\Sigma_{1}}+ \underbrace{\mathsf{B}^{(\rho({x}))}\Big(\mathsf{W}^{\mathsf{SRPT}_{1,N}}_{\leq x}(r_{x})\Big)}_{\Sigma_{2}}.
\end{align*}
In $(a)$ we use  the upper bounds established in (\ref{wastebound}) and (\ref{rele_bound}); $(b)$ follows from the additivity of busy period in $\mathsf{M}$/$\mathsf{G}$/$1$.

\color{black}
\paragraph{{Asymptotic Analysis of the Bound}} We analyze the expectations of \(\Sigma_1\) and \(\Sigma_2\). The term \(\Sigma_1\) represents delay due to non-preemptivity and multi-server setup:
\color{black}
\begin{align}\label{ineqbound0}
\mathbbm{E}_{x,r_{x}}[\Sigma_{1}]&= O\Big(\mathbbm{E}\Big(\mathsf{B}^{(\rho(x))}(\eta+x)\Big)\Big)= O\Big(\mathbbm{E}\Big[\frac{\eta+x}{1-\rho(x)}\Big] \Big)\notag\\
&=O\Big(\log \frac{1}{1-\rho}\Big)+\mathbbm{E}[\eta]\cdot O\Big(\int_{0}^{\infty}\frac{{ f(x)}}{1-\rho(x)}\mathrm{d}x \Big).
\end{align}
{\color{black}{
Splitting the integral at \(\xi\) such that \(\rho(\xi) = \rho/2\):}}
\begin{align}\label{ineqbound}
\int_{0}^{\infty}\frac{f(x)}{1-\rho(x)}\mathrm{d}x =&\int_{0}^{\xi}\frac{f(x)}{1-\rho(x)}\mathrm{d}x + \int_{\xi}^{\infty}\frac{f(x)}{1-\rho(x)}\mathrm{d}x\notag\\
\leq &\frac{1}{1-\rho(\xi)} +\frac{1}{\xi}\cdot \int_{\xi}^{\infty}\frac{xf(x)}{1-\rho(x)}\mathrm{d}x.
\end{align}
Note that
\begin{align*}
\rho(\xi)=\lambda \cdot \int_{0}^{\xi}{tf(t)\mathrm{d}t}\leq \lambda \cdot\xi,
\end{align*}
we have $\xi \geq \mathbbm{E}[p_{i}]/2$. {\color{black}{Substituting into (\ref{ineqbound}):}}
\begin{align*}
\int_{0}^{\infty}\frac{f(x)}{1-\rho(x)}\mathrm{d}x  \leq &2+ \frac{2}{\mathbbm{E}[p_{i}]} \cdot  \int_{0}^{\infty}\frac{xf(x)}{1-\rho(x)}\mathrm{d}x\\
=& 2+ \frac{2}{\mathbbm{E}[p_{i}]} \cdot \log \frac{1}{1-\rho}.
\end{align*}
Thus, 
\begin{align*}
\mathbbm{E}_{x,r_{x}}{[\Sigma_{1}]} = O\Big( \log \frac{1}{1-\rho}\Big)\cdot \mathbbm{E}[\eta]
\end{align*}

\color{black}

{\color{black}{
The term \(\Sigma_2\) is the busy period generated by higher-priority work in \(\mathsf{SRPT}\) as seen by an arriving job. By the PASTA property \cite{wolff1982poisson}:}}
\begin{align*}
\mathbbm{E}{[F_{\rho}^{\mathrm{SRPT}_{1, N}}]}\geq  &\mathbbm{E}_{x,r_{x}}{[\mathsf{B}^{(\rho(x))}(W^{\mathrm{SRPT}_{1, N}}_{\leq x}(r_{x}))]}\\
=& \mathbbm{E}_{x,r_{x}}{[\Sigma_{2}]},
\end{align*}

\color{black}
\paragraph{Putting the Pieces Together}
Combining the expectations of \(\Sigma_1\) and \(\Sigma_2\),
\color{black}
\begin{align}\label{unrefine}
\mathbbm{E}{[F_{\rho}^{ \mathsf{NP\text{-}SRPT}}]}&=\mathbbm{E}_{x,r_{x}}{[\mathsf{T}^{ \mathsf{NP\text{-}SRPT}}_{x}]}=\mathbbm{E}_{x,r_{x}}[\Sigma_{1}]+\mathbbm{E}_{x,r_{x}}[\Sigma_{2}]\notag\\
&\leq \mathbbm{E}{[F_{\rho}^{\mathrm{SRPT}_{1, N}}]}+O\Big( \log \frac{1}{1-\rho}\Big)\cdot \mathbbm{E}[\eta].
\end{align}

This completes the proof.

\end{proof}

\subsection{Optimality of $\mathsf{NP}$-$\mathsf{SRPT}$ with finite task size}
The benchmark system we consider consists of a single machine with speed $N$, where all the tasks can be allowed to be served in preemptive fashion, \ie, the concept of task is indeed unnecessary in this setting. It is clear to see that the mean response time under optimal algorithm for this single machine system can be performed as a valid lower bound for the multitask problem, \ie, 
\begin{align}\label{ineqq1}
\mathbbm{E}{[F^{\pi^{*}}_{\rho}]}\geq \mathbbm{E}{[F^{\mathsf{SRPT}_{1,N}}_{\rho}]}.
\end{align}
We leverage the following conclusion to bound the heavy traffic growth rate of the average response time under SRPT. 


\begin{lemma}[\hspace*{-0.3em}\cite{scully2020gittins}] \label{mg1bound}
If $\mathbbm{E}[p^{2}(\log p)^{+}]<\infty$, 
\begin{align*}
\mathbbm{E}[F^{\mathrm{SRPT}_{1,1}}_{\rho}]=\omega\Big(\log \frac{1}{1-\rho}\Big).
\end{align*}
\end{lemma}


\paragraph{Proof of optimality} It suffices to show that the difference between the average response time under $\mathsf{NP}$-$\mathsf{SRPT}$  and the optimal algorithm is a lower order term, \ie,
\begin{align}\label{optcondi}
\lim_{n \rightarrow \infty}\frac{\mathbbm{E}{[F_{\rho}^{ \mathsf{NP\text{-}SRPT}}]}-\mathbbm{E}[F^{\mathsf{SRPT}_{1,N}}_{\rho}]}{\mathbbm{E}[F^{\mathsf{SRPT}_{1,N}}_{\rho}]}=0,
\end{align}
which holds according to Lemma~\ref{mg1bound} and inequality (\ref{unrefine})-(\ref{ineqq1}).


\section{Optimality of $\mathsf{NP}$-$\mathsf{SRPT}$ Beyond Finite Task Size}\label{secdis}
Up to this point, we have focused on job size distributions with finite task, which is rather restrictive. It is natural to consider relaxations of this assumption.  In this section, we turn to other classes of job size distributions with unbounded support on task workload. These results provide complement to our developments about the theory of the asymptotic optimality of $\mathsf{NP}$-$\mathsf{SRPT}$ .

\subsection{Warm up--random number of tasks}

In the following proposition, we analyze the order of the expected value of the maximum task size, if 
\begin{enumerate}[label=(\alph*), ref=(\alph*)]
    \item\label{item:a}  the moment generating function of the job size distribution exists.
    \item\label{item:b} the $k$-th moment of job size is finite.
\end{enumerate}

\begin{proposition}\label{promax}
For $n$ jobs with independently distributed workload, the mean value of the maximum task size is in the order of,

\begin{equation}
    \mathbbm{E}[\eta] = \begin{cases}
        O(\log (\mathbbm{E}[n])), & \text{case \ref{item:a}} , \\
        O((\mathbbm{E}[n])^{1/k}), & \text{case \ref{item:b}}.
    \end{cases}
\end{equation}
\end{proposition}
\begin{proof}
Consider an increasing convex function $g(x)$, by Jensen's inequality, we have
\begin{align*}
g(\mathbbm{E}[\eta])\leq \mathbbm{E}[g(\eta)] \leq \mathbbm{E}[\sum_{i=1}^{n}g(p_{i})] = \mathbbm{E}[n]\mathbbm{E}[g(p_{i})].
\end{align*}
Since the $g^{-1}(\cdot)$ is also increasing, hence
\begin{align*}
\mathbbm{E}[\eta]  \leq g^{-1}(\mathbbm{E}[n]\mathbbm{E}[g(p_{i})]).
\end{align*}
Let $s>0$ and 
\begin{equation*}
    g(x) = \begin{cases}
        e^{sx}, & \text{case \ref{item:a}} , \\
        x^{k}, & \text{case \ref{item:b}}.
    \end{cases}
\end{equation*}
then we have
\begin{equation}
    \mathbbm{E}[\eta] \leq \begin{cases}
        \min_{s\in D(s)}\frac{\log (\mathbbm{E}[n]) + \log(\mathbbm{E}[e^{sp_{i}}])}{s} , & \text{case \ref{item:a}} , \\
        (\mathbbm{E}[n])^{1/k}\cdot \mathbbm{E}^{1/k}[p^{k}_{i}], & \text{case \ref{item:b}}.
    \end{cases}
\end{equation}
where $D(s)=\{s>0|\mathbbm{E}[e^{sp_{i}}]<\infty\}$. The proof is complete.
\end{proof}

\subsection{Optimality in $\mathsf{M}$/$\mathsf{M}$/$\mathsf{1}$}

For single server with Poisson arrival and 
exponentially distributed workload, we claim that $\mathsf{NP}$-$\mathsf{SRPT}$  is asymptotic optimal without any additional assumptions. 
\begin{theorem}\label{mm1opt}
 $\mathsf{NP}$-$\mathsf{SRPT}$ is asymptotic optimal in $\mathsf{M}$/$\mathsf{M}$/$\mathsf{1}$.
\end{theorem}
We first introduce the following propositions that will be used in our proof.

\begin{proposition}[\hspace*{-0.6em}\cite{bansal2018achievable}]\label{pro2}
For $\mathsf{M}$/$\mathsf{M}$/$\mathsf{1}$ model and any work-conserving algorithm, let $n_{\mathrm{busy}}$ be the number of arrivals in a busy period, then 
\begin{align*}
\mathbbm{E}[n_{\mathrm{busy}}]=O\Big(\frac{1}{1-\rho}\Big).   
\end{align*}
\end{proposition}

\begin{proofof}{Theorem \ref{mm1opt}}
Firstly, it has been shown in \cite{DBLP:journals/orl/Bansal05} that, the average response time under SRPT in $\mathsf{M}$/$\mathsf{M}$/$\mathsf{1}$ satisfies
\begin{align*}
\frac{\mathbbm{E}[F^{\mathrm{SRPT-1}}_{\rho}]}{\mu(1-\rho)\log (1/(1-\rho))}\in \Big[\frac{1}{18e}, 7\Big]
\end{align*}
for $\rho \in [2/3,1)$. Based on Proposition \ref{promax} and Proposition \ref{pro2}, we have
\begin{align*}
\mathbbm{E}[\eta] =O(\log(\mathbbm{E}[n_{\mathrm{busy}}])) = O\Big(\log \frac{1}{1-\rho}\Big).  
\end{align*}
Combined with Theorem \ref{lemma1}, this implies that
\begin{align*}
\mathbbm{E}{[F_{\rho}^{ \mathsf{NP\text{-}SRPT}}]}- \mathbbm{E}{[F_{\rho}^{\mathsf{SRPT}_{1,N}}]} \leq O\Big(\log^{2} \frac{1}{1-\rho}\Big), 
\end{align*}
which is a lower order term compared to $\mathbbm{E}[F^{\mathrm{SRPT-1}}_{\rho}]$.
\end{proofof}

\subsection{Beyond Exponential Job Size Distribution}


In the study of $\mathsf{NP}$-$\mathsf{SRPT}$, it is crucial to analyze the expected length of busy periods, which could help bound $\mathbbm{E}[\eta]$. Notably, previous works have established insights into the analysis of busy periods in $\mathsf{M}$/$\mathsf{M}$/$\mathsf{N}$ queuing systems\cite{omahen1978analysis, artalejo2001analysis}. In this section, we extend the analysis in \cite{omahen1978analysis} to $\mathsf{M}$/$\mathsf{G}$/$\mathsf{N}$ and prove an upper bound on the expected length of the busy period under $\mathsf{NP}$-$\mathsf{SRPT}$, with the following assumption. 

\begin{assumption}\label{residualass}
For any $a\geq 0$, there exist constants $r_{\max}$ such that  the expected size of the residual job, \ie, the amount of work remaining after a certain point in time or after a certain amount of work has been completed, satisfies 
\begin{align*}
\mathbbm{E}[p-a|p>a]\leq r_{\max}.
\end{align*}
\end{assumption}

In simpler terms, this assumption states that regardless of the current finished workload, the expected size of the remaining work will always be no more than $r_{\max}$, ie, being finite. It actually generalizes the concept of job size characterized by the exponential and \emph{new better than used} (NBUE) distribution.

\begin{lemma}
Under Algorithm $\mathsf{NP}$-$\mathsf{SRPT}$  and Assumption \ref{residualass}, 
\begin{align*}
\mathbbm{E}[\eta] = \begin{cases}
O(\log (1/(1-\rho))), &\text{case \ref{item:a}}\\
O((1-\rho)^{-\frac{1}{k-1}}), &\text{case \ref{item:b}}
\end{cases}  
\end{align*}
\end{lemma}

\begin{proof}
To analyze the busy periods under the $\mathsf{NP}$-$\mathsf{SRPT}$  policy, we classify two types of busy periods in an $N$-server system:
\begin{itemize}
\item \emph{Full busy period}. A time interval during which all $k \leq N$ servers are occupied. It starts when a new arrival finds $k-1$ customers already in the system and ends at the first departure epoch when exactly $k-1$ customers remain in the system. The length of such a busy period is denoted by $T_k$.
\item \emph{Partial busy period}. A time interval during which at least one server is busy. It begins when a new arrival finds the system empty and ends at the first epoch when the system becomes empty again.
\end{itemize}
Conditioning on whether a departure or a new arrival occurs first, the expected length of a full busy period can be expressed as:
\begin{align}\label{condition_pro}
&\mathbbm{E}[T_{k}|\text{start with} \{\mathsf{J}_{i}\}_{i=1}^{k}, p_{\mathsf{J}_{i}}\leq p_{\max}]\notag\\
=&\mathbbm{P}(\text{departure first})\cdot \mathbbm{E}[\text{first departure time}|\text{departure first}]\notag\\
&+ \mathbbm{P}(\text{arrival first})\cdot \mathbbm{E}[T_{k+1}|\text{start with} \{\mathsf{J}^{\prime}_{i}\}_{i=1}^{k}\cup\{\mathsf{J}_{k+1}\}]\notag\\
&+ \mathbbm{P}(\text{arrival first})\cdot \mathbbm{E}[T_{k}|\text{start with} \{\mathsf{J}^{\prime\prime}_{i}\}_{i=1}^{k}\}],
\end{align}
for any initial job collection $\{\mathsf{J}_{i}\}_{i=1}^{k}$.

Since we consider an $\mathsf{M}$/$\mathsf{G}$ system, where job arrivals follow a Poisson process, the probability that an arrival occurs before a departure is:
\begin{align*}
&\mathbbm{P}(\text{arrival first})\\
=&\int_{0}^{\infty}\mathbbm{P}(\text{minimum remaining workload}\\
&\text{of jobs under processing}\geq s)\cdot \lambda e^{-\lambda s} \mathrm{d}s\\
\leq & \int_{0}^{2r_{\max}} \lambda e^{-\lambda s} \mathrm{d}s + \frac{1}{2}\int_{2r_{\max}}^{\infty}{\lambda e^{-\lambda s} \mathrm{d}s}\\
= & 1-\frac{1}{2}e^{-2\lambda r_{\max}} \leq 1-\frac{1}{2}e^{-\frac{2r_{\max}^{2}}{N}},
\end{align*}
where the second inequality follows from Markov inequality and Assumption \ref{residualass}.

Now, define the worst-case expected length of a full busy period starting with $k$ jobs and job sizes bounded by $p_{\max}$:
\begin{align*}
L_{k}(p_{\max}) = \sup_{\{\mathsf{J}_{i}\}_{i=1}^{k}} \mathbbm{E}[T_{k}|\text{start with} \{\mathsf{J}_{i}\}_{i=1}^{k}, p_{\mathsf{J}_{i}}\leq p_{\max}].
\end{align*}
Using (\ref{condition_pro}), we derive
\begin{align*}
& L_{k}(p_{\max})\\
\leq &\frac{1}{\mu} + \Big(1-\frac{1}{2}e^{-\frac{2r_{\max}^{2}}{N}}\Big)(L_{k+1}(p_{\max})+L_{k}(p_{\max}))\\
\Rightarrow & L_{k}(p_{\max}) \leq \frac{2e^{\frac{2r_{\max}^{2}}{N}}}{\mu}+\Big(2e^{\frac{2r_{\max}^{2}}{N}}-1\Big)L_{k+1}(p_{\max}),
\end{align*}
and thus
\begin{align*}
L_{k}(p_{\max}) \leq \frac{2e^{\frac{2r_{\max}^{2}}{N}}}{\mu}+\Big(2e^{\frac{2r_{\max}^{2}}{N}}-1\Big)L_{k+1}(p_{\max}).
\end{align*}
This recursion implies that $L_k(p_{\max}) = O(L_N(p_{\max}))$ for all $k \leq N$, given a fixed number of servers.

Since $L_{N}(p_{\max})$ is no more than the busy period of a single server with the same arrivals and initial workload of $Np_{\max}$, we obtain
\begin{align*}
L_{k}(p_{\max}) = O\Big(\frac{p_{\max}}{1-\rho}\Big), \forall k \leq N.
\end{align*}
Let $\mathbb{E}[\mathsf{B}^{(N)}_{\mathsf{NP\text{-}SRPT}}]$ denote the expected duration of a partial busy period, defined as the time from the first job arrival in an empty $N$-server system until the system becomes empty again under the $\mathsf{NP}$-$\mathsf{SRPT}$ policy. Then
\begin{align*}
\mathbbm{E}[\mathsf{B}^{(N)}_{ \mathsf{NP\text{-}SRPT}}] \leq &\mathbbm{E}_{p_{\max}}\Big[\sum_{k=1}^{N} {L_{k}(p_{\max})}\Big] \\
= & \begin{cases}
        O(\frac{\log (\mathbbm{E}[\mathsf{B}^{(N)}_{ \mathsf{NP\text{-}SRPT}}])}{1-\rho}), & \text{case \ref{item:a}} , \\
        O(\frac{(\mathbbm{E}[\mathsf{B}^{(N)}_{ \mathsf{NP\text{-}SRPT}}])^{1/k}}{1-\rho}), & \text{case \ref{item:b}}.
    \end{cases}
\end{align*}
Solving this inequality yields:
\begin{align*}
\mathbbm{E}[\mathsf{B}^{(N)}_{ \mathsf{NP\text{-}SRPT}}] &= 
\begin{cases}
        O(\frac{1}{(1-\rho)^{1+\epsilon}}), & \text{case \ref{item:a}} , \\
        O((\frac{1}{1-\rho})^{\frac{k}{k-1}}), & \text{case \ref{item:b}}.
    \end{cases}
\end{align*}
This completes the proof, following the arguments in Proposition~\ref{promax}.
\end{proof}

\paragraph{Additional lower bounds on optimal response time.} In addition to exponential distribution, Lin et al. \cite{lin2011heavy} also gave a characterization of the heavy-traffic behavior of SRPT with general job size distribution. We first introduce the concept of Matuszewska index, which plays a significant role in the result.

\begin{definition}[Upper Matuszewska Index \cite{lin2011heavy}]
Let $f$ be a positive function defined in $[0,\infty)$, the upper Matuszewska index is defined as the infimum of $\alpha$ for which there exists a constant $C=C(\alpha)$ such that for each $\bar{\lambda}>1$, 
\begin{align*}
\lim\nolimits_{x\rightarrow \infty}\frac{f(\lambda x)}{f(x)}\leq C\lambda^{\alpha},
\end{align*}
holds uniformly for $\lambda\in [1, \bar{\lambda}]$.
\end{definition}
This index helps us understand the asymptotic behavior of function $f(\cdot)$.

\begin{proposition}[\hspace*{-0.6em}\cite{lin2011heavy}] \label{flowboundfact}In an $\mathsf{M}$/$\mathsf{G}$/$1$ queue, if the upper Matuszewska index of the job size distribution is less than $-2$, then
\begin{align*}
\mathbbm{E}{[F^{\mathrm{SRPT-1}}_{\rho}]}=\Theta \Big(\frac{1}{(1-\rho)\cdot G^{-1}(\rho)}\Big),
\end{align*}
where $G^{-1}(\cdot)$ denotes the inverse of $G(x)=\rho_{\leq x}/\rho=\int_{0}^{x}{tf(t) \mathrm{d}t}/\mathbbm{E}[p_{i}]$.
\end{proposition}
For example, exponential distribution has an upper Matuszewska index $M_{f}=-\infty$ and $G^{-1}(\rho)=\Theta(\log (1/(1-\rho)))$.
\begin{theorem}\label{generalthem}
The average response time under $\mathsf{NP}$-$\mathsf{SRPT}$  is asymptotic optimal in $\mathsf{M}$/$\mathsf{G}$/$\mathsf{N}$ under Assumption \ref{residualass}, if upper Matuszewska index of job size distribution is less than $-2$ and
\begin{align*}
G^{-1}(\rho)= \begin{cases}
o\Big(\frac{1}{(1-\rho)\cdot \log^{2}(1/(1-\rho))}\Big), &\text{case \ref{item:a}}\\
o\Big(\frac{1}{(1-\rho)^{1-1/k-\epsilon}\cdot \log(1/(1-\rho))}\Big), & \text{case \ref{item:b}}
\end{cases}
\end{align*}
\end{theorem}
Theorem \ref{generalthem} identifies sufficient conditions under which the $\mathsf{NP}$-$\mathsf{SRPT}$  scheduling policy achieves asymptotic optimality in terms of average response time. Examples of the class of distribution in Theorem \ref{generalthem} include but not limited to \emph{Weibull distribution}, \emph{Pareto distribution} and \emph{regularly varying distributions} \cite{lin2011heavy}. 

\subsection{Discussions}
In our analysis of $\mathsf{NP\text{-}SRPT}$, we consider a scenario where tasks within the same job cannot be executed in parallel, necessitating sequential processing. This setting naturally accommodates precedence constraints, such as those represented by a Directed Acyclic Graph (DAG), where tasks must follow a specific order (e.g., a topological sort). In this case, $\mathsf{NP\text{-}SRPT}$ remains effective by scheduling only the next ready task (with all predecessors completed) from the job with the shortest remaining processing time, respecting the topological order. This ensures that precedence constraints are satisfied while maintaining the greedy prioritization of $\mathsf{NP\text{-}SRPT}$. When tasks can be processed in parallel, finding the optimal allocation policy or analyzing the mean job response time for different allocation policies remains a challenging open problem \cite{harchol2021open}. 

\color{black}

\section{Extensions to Scheduling with Unknown Job Sizes}\label{sec:unknown}

The preceding analysis has established the asymptotic optimality of $\mathsf{NP}$-$\mathsf{SRPT}$ under the assumption that job sizes are known upon arrival.
However, in many practical systems, this information is not available to the scheduler \cite{scully2021optimal, wang2022approximate}.
This motivates the question of how to design effective scheduling policies when job sizes are unknown.
In this section, we address this challenge by extending the analysis to unknown job size setting.
We build upon the work of Scully et al.~\cite{scully2021optimal}, and adapt their insights to the non-preemptive, multi-task environment.


We specifically consider non-preemptive version of the SOAP policy, \emph{monotone shortest expected remaining processing time} ($\mathsf{M}$-$\mathsf{SERPT}$) and \emph{monotone Gittins}  ($\mathsf{M}$-$\mathsf{Gittins}$) policies, denoted as $\mathsf{NP}$-$\mathsf{M}$-$\mathsf{SERPT}_N$ and $\mathsf{NP}$-$\mathsf{M}$-$\mathsf{Gittins}_N$ respectively. The results for these policies apply to a wide range of job size distributions,  categorized by their tail behavior. The key classes include O-Regularly Varying ($\mathsf{OR}$), Quasi-Decreasing Hazard Rate ($\mathsf{QDHR}$), and distributions in the Gumbel Domain of Attraction ($\mathsf{MDA}(\Lambda)$) \cite{scully2021optimal}. For the sake of brevity and to maintain focus on the main results, the formal mathematical definitions for these classes are provided in Appendix \ref{app:definitions}.

The following lemma relates the performance of these non-preemptive multi-server policies to their single-server preemptive counterparts in the heavy-traffic regime.

\begin{lemma}
Let $\pi \in \{\mathsf{M\text{-}SERPT}, \mathsf{M\text{-}Gittins}\}$ and consider a job size distribution in $\mathsf{OR}(-\infty, -1)\cup (\mathsf{QDHR}\cap \mathsf{MDA}(\Lambda))$, then
\begin{align}\label{eq:limit}
\lim\nolimits_{\rho \to 1} \frac{\mathbbm{E}[F^{\mathsf{NP}\text{-}\pi_N}_{\rho}]}{\mathbbm{E}[F^{\pi_1}_{\rho}]} = 1.
\end{align}
\end{lemma}

\begin{proof}
Since $\pi \in \{\mathsf{M\text{-}SERPT}, \mathsf{M\text{-}Gittins}\}$ is a monotone non-decreasing SOAP policy, we can apply Lemma~\ref{generallemma} to bound the difference in relevant work for job with size $x$ between the non-preemptive $N$-server policy $\mathsf{NP}\text{-}\pi$ and its preemptive single-server counterpart $\pi_{1, N}$:
\begin{align}\label{np-pi-n}
\mathsf{Rel}^{\mathsf{NP}\text{-}\pi}_{\leq x}(t) - \mathsf{Rel}^{\pi_{1, N}}_{\leq x}(t) \leq (N-1) (z^{\pi_{1, N}}_{x} + \eta),
\end{align}
where $z_{x}^{\pi_{1, N}}$ represents the value of $z^{\pi_{1, N}}_{\mathsf{J}}(r^{\pi_{1, N}}(p_{\mathsf{J}}, p_{\mathsf{J}}))$ for job $\mathsf{J}$ with size $x$. Following an argument analogous to that in \cite{scully2021optimal}, inequality~\eqref{np-pi-n} implies a bound on the expected response time:
\begin{align}\label{flow_bound}
&\mathbbm{E}[F^{\mathsf{NP}\text{-}\pi}_{\rho}]-\mathbbm{E}[F^{\pi_{1, N}}_{\rho}]\notag\\
\leq & o(\mathbbm{E}[F^{\pi_{1, N}}_{\rho}]) + O\left(\mathbbm{E}\left[\frac{\eta}{1-\rho(y^{\pi_{1, N}}_{x})}\right]\right),
\end{align}
where $y^{\pi_{1, N}}_{x}$ denotes the \emph{new age cutoff}. Let us define a threshold $\xi$ such that $\rho(y^{\pi_{1, N}}_{\xi}) = \rho/2$, this implies:
\[
\frac{\rho}{2} = \int_{0}^{y^{\pi_{1, N}}_{\xi}} \lambda \bar{F}(t) \, \mathrm{d}t \leq \lambda y^{\pi_{1, N}}_{\xi}.
\]
As established in \cite{scully2021optimal} for $\pi \in \{\mathsf{M\text{-}SERPT}, \mathsf{M\text{-}Gittins}\}$, the age cutoff for the single-server policy, $y^{\pi_{1, N}}_{x}$, exhibits the following asymptotic behavior:
\begin{itemize}
    \item $y^{\pi_{1, N}}_{x} = \Theta(x)$ for job size distributions in $\mathsf{OR}(-\infty, -1)$. 
    \item $y^{\pi_{1, N}}_{x} = \Omega(x^{1/\gamma})$ for some $\gamma \geq 1$, for distributions in $\mathsf{QDHR}\cap \mathsf{MDA}(\Lambda)$. 
\end{itemize}
In both scenarios, this yields a lower bound on $\xi$:
\[
\xi = \Omega\left(\left(\frac{\rho}{2\lambda}\right)^{\gamma}\right)
\]
for some constant $\gamma \geq 1$. Consequently, we can bound the final term in \eqref{flow_bound} as, 
\begin{align*}
&\mathbbm{E}\left[\frac{\eta}{1-\rho(y^{\pi_{1, N}}_{x})}\right] \\
\leq & \left(\frac{1}{1-\rho(y^{\pi_{1, N}}_{\xi})} + \frac{1}{\xi} \mathbbm{E}\left[\frac{x}{1-\rho(y^{\pi_{1, N}}_{x})}\right]\right) \cdot \mathbbm{E}[\eta] \\
\leq & O\left((1-\rho)^{-\varepsilon}\right) \cdot \mathbbm{E}[\eta], \quad \forall \varepsilon > 0.
\end{align*}
Since $\mathbbm{E}[F^{\pi_{1, N}}_{\rho}]$ grows at least polynomially in $(1-\rho)^{-1}$, this error term is asymptotically negligible .
Substituting this result into \eqref{flow_bound} yields the desired limit.
\end{proof}


The lemma establishes that in the heavy-traffic limit, the mean response time under the non-preemptive, $N$-server policies $\mathsf{NP}$-$\mathsf{M}$-$\mathsf{SERPT}_N$ and $\mathsf{NP}$-$\mathsf{M}$-$\mathsf{Gittins}_N$ converges to that of their single-server, preemptive counterparts.
This allows us to leverage known results for the single-server $\mathsf{M}$/$\mathsf{G}$/$1$ setting.
Specifically, the Gittins index policy is optimal, and it is known that $\mathsf{M}$-$\mathsf{Gittins}$ is asymptotically optimal (i.e., $\lim_{\rho \to 1} \mathbbm{E}[F^{\mathsf{M\text{-}Gittins}_1}_{\rho}] / \mathbbm{E}[F^{\mathsf{Gittins}_1}_{\rho}] = 1$) and $\mathsf{M}$-$\mathsf{SERPT}$ is 2-competitive (i.e., $\mathbbm{E}[F^{\mathsf{M\text{-}SERPT}_1}_{\rho}] \leq 2\mathbbm{E}[F^{\mathsf{Gittins}_1}_{\rho}]$).
These facts lead to our main theorem for the unknown job size case.

\begin{theorem}
\label{unkonwn_job_theorem}
In the heavy-traffic limit ($\rho \to 1$) for a multi-task multi-server system with unknown job sizes from a distribution in $\mathsf{OR}(-\infty, -1)\cup (\mathsf{QDHR}\cap \mathsf{MDA}(\Lambda))$:
\begin{enumerate}
    \item[(i)] The $\mathsf{NP}$-$\mathsf{M}$-$\mathsf{Gittins}_N$ policy is asymptotically optimal.
    \item[(ii)] The $\mathsf{NP}$-$\mathsf{M}$-$\mathsf{SERPT}_N$ policy is asymptotically $2$-competitive.
\end{enumerate}
\end{theorem}

\color{black}

\section{Experimental Results}\label{exp}

To evaluate the asymptotic optimality of $\mathsf{NP}$-$\mathsf{SRPT}$, we conducted experiments using a Weibull job size distribution. Weibull distribution has a cumulative distribution function of $F(x)=1-e^{-\mu x^{k}}$, upper Matuszewska index $M_{f}=-\infty$ and $G^{-1}(\rho)=\Theta({(\log (1/(1-\rho))}^{1/k})$. Indeed exponential distribution is a special case of the Weilbull distribution with $\alpha=1$. Jobs were randomly divided into $2$ to $5$ non-preemptive tasks.

As shown in Figure \ref{fig:image}, we analyzed the ratio, calculated as the average response time of $\mathsf{NP}$-$\mathsf{SRPT}$ divided by the average response time of the single server $\mathsf{SRPT}$. The experiments were performed for two different parameters, $k=1$, which reduces to an exponential distribution, and $k=10$. The average job size is set to be $1$, so $\lambda = 1/\Gamma(1+1/k)$. As $\rho$ approaches 1, the ratio converges to 1 for both $k=1$ and $k=10$. This convergence empirically validates the asymptotic optimality of $\mathsf{NP}$-$\mathsf{SRPT}$. Notably, the convergence rate for $k=10$ is faster compared to $k=1$. This observation aligns with the convergence characteristics of the theoretical bound, since as the value of $k$ increases, the expected value of $\eta$ diminishes, and the average response time under single-server SRPT increases.

\begin{figure}[H]
  \centering
  \includegraphics[width=\linewidth]{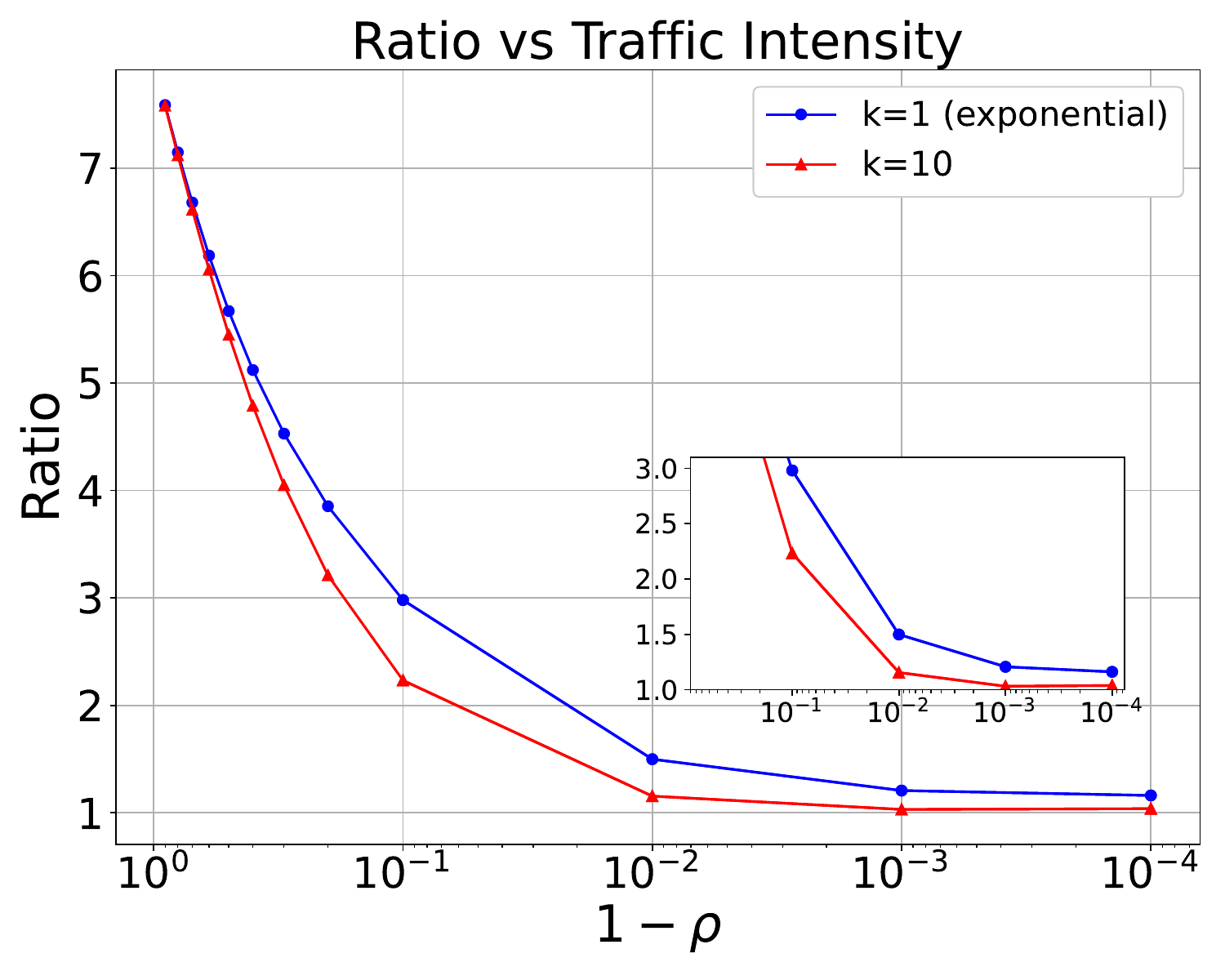}
  \caption{Convergence of ratio with respect to traffic intensity}
  \label{fig:image}
\end{figure}

\section{Conclusion}\label{seccon}
In this work, we study the multitask scheduling  problem, for which the optimal algorithms and tight analyses remain widely open for almost all settings. We propose $\mathsf{NP}$-$\mathsf{SRPT}$ algorithm, which achieves a competitive ratio that is order optimal when the number of machines is constant. Another appealing and more important property of $\mathsf{NP}$-$\mathsf{SRPT}$ is that, the average response time incurred under Poisson arrival is asymptotic optimal when the traffic intensity goes to $1$, if task service times are finite or job size distribution satisfies some mild conditions. {\color{black}Furthermore, by extending the analysis to scenarios with unknown job sizes, we demonstrate the robustness of our approach and establish the asymptotic optimality of policies for such settings.}

\section*{Acknowledgments}  
The author thanks Ziv Scully for his many invaluable suggestions, including simplifying the proof for bounding the remaining workload under NP-SRPT and the usage of the WINE identity, and for his feedback on improving the
presentation of earlier versions of this article. The author also
thanks the anonymous reviewers for their valuable comments and suggestions.
\bibliography{scheduling}
\bibliographystyle{unsrt}

\appendix

\section*{Proof of Lemma \ref{num_of_job_bound}}\label{appendixcomproof}
\begin{proof}
In the following of the proof, we use $\mathcal{C}^{[k]}(\pi, t)=\cup_{i=1}^{k}\mathcal{C}_{i}(\pi, t)$ to denote the collection of jobs in the first $k$ classes, and let $W_{\pi}^{[k]}(t)=\sum_{i=1}^{k} {W_{\pi}^{(i)}(t)}$ represent the total remaining workload of jobs in the first $k$ classes, where $W_{\pi}^{(k)}(\pi, t)$ denotes the amount of remaining workload of jobs in class $\mathcal{C}_k(\pi,t)$. $W_{\pi^{*}_{1,N}}^{(k)}(t)$ and $W_{\pi^{*}_{1,N}}^{[k]}(t)$ are defined in a similar way for $\pi^{*}_{1,N}$. Without loss of generality, we assume $\log p_{\max}$ and $\log p_{\min}$ are integers.

For $\forall t\geq 0$, the number of unfinished jobs under the optimal algorithm is no less than,
\begin{align}
n_{\pi^{*}_{1,N}}(t) \ge & \sum_{k=\log p_{\min}}^{\log p_{\max}+1}{\frac{W_{\pi^{*}_{1,N}}^{(k)}(t)}{2^{k}}}\notag\\
= &\sum_{k=\log p_{\min}}^{\log p_{\max}+1}{\frac{\Big[W_{\pi^{*}_{1,N}}^{[k]}(t)-W_{\pi^{*}_{1,N}}^{[k-1]}(t)\Big]}{2^{k}}}\notag\\
=&\frac{W_{\mathrm{\pi^{*}_{1,N}}}^{[\log p_{\max}+1]}(t)}{2^{\log p_{\max}+1}}+\sum_{k=\log p_{\min}}^{\log p_{\max}+1}{\frac{W_{\mathrm{\pi^{*}_{1,N}}}^{[k]}(t)}{2^{k+1}}}\notag\\
\geq& \sum_{k=\log p_{\min}}^{\log p_{\max}+1}{\frac{W_{\pi^{*}_{1,N}}^{[k]}(t)}{2^{k+1}}}\label{jobaliveoptimal}.
\end{align}
On the other hand, the number of jobs alive under $\mathsf{NP}$-$\mathsf{SRPT}$  can be upper bounded in a similar fashion,
\begin{align*}
n_{ \mathsf{NP\text{-}SRPT}}(t)\leq & \sum_{k=\log p_{\min}}^{\log p_{\max}+1}{\frac{W_{ \mathsf{NP\text{-}SRPT}}^{(k)}(t)}{2^{k-1}}}\\
=&\sum_{k=\log p_{\min}}^{\log p_{\max}+1}{\frac{\Big[W_{ \mathsf{NP\text{-}SRPT}}^{[k]}(t)-W_{ \mathsf{NP\text{-}SRPT}}^{[k-1]}(t)\Big]}{2^{k-1}}}\\
=&\sum_{k=\log p_{\min}}^{\log p_{\max}}{\frac{W_{ \mathsf{NP\text{-}SRPT}}^{[k]}(t)}{2^{k}}}+\frac{W_{ \mathsf{NP\text{-}SRPT}}^{[\log p_{\max}+1]}(t)}{2^{\log p_{\max}}}\\
\leq & \sum_{k=\log p_{\min}}^{\log p_{\max}+1}{\frac{W_{ \mathsf{NP\text{-}SRPT}}^{[k]}(t)}{2^{k-1}}}.
\end{align*}
Using Lemma~\ref{remaining_workload}, we are able to relate the number of unfinished jobs under two algorithms,
\begin{align*}
n_{ \mathsf{NP\text{-}SRPT}}(t) \leq& \sum_{k=\log p_{\min}}^{\log p_{\max}+1}{\frac{W_{ \mathsf{NP\text{-}SRPT}}^{[k]}(t)}{2^{k-1}}}\\
\leq &\sum_{k=\log p_{\min}}^{\log p_{\max}+1}{\frac{W_{\pi^{*}_{1,N}}^{[k]}(t)}{2^{k-1}}}+\sum_{k=\log p_{\min}}^{\log p_{\max}+1}{\frac{N\cdot (2^{k}+\eta)}{2^{k-1}}}\\
\leq &4n_{\pi^{*}_{1,N}}(t)+4N\cdot \Big(\log \alpha+\beta+1\Big),
\end{align*}
where the last inequality follows from inequality (\ref{jobaliveoptimal}). 
\end{proof}

\color{black}
\section*{Formal Definitions of Job Size Distribution Classes}
\label{app:definitions}

Our analytical results for unknown job size setting in Theorem \ref{unkonwn_job_theorem}, apply to specific families of job size distributions that are formally defined as follows.

\begin{definition}[Job Size Distribution Classes \cite{scully2021optimal}]
The relevant classes are characterized by their asymptotic tail behavior.
\begin{description}
    \item[\textbf{O-Regularly Varying ($\mathsf{OR}$)}] This class describes distributions with heavy, power-law-like tails. A function $f$ is considered O-regularly varying if for all sufficiently large values (i.e., for all $y \ge x \ge x_{0}$), its growth is bounded above and below by power laws. This requires the existence of constants $C_{0}, x_{0} > 0$ and exponents $\beta \ge \alpha > 0$ such that:
    $$
    \frac{1}{C_{0}}\left(\frac{y}{x}\right)^{-\beta} \le \frac{f(y)}{f(x)} \le C_{0}\left(\frac{y}{x}\right)^{-\alpha}
    $$
    A distribution is classified as $\mathsf{OR}(-\beta_{0}, -\alpha_{0})$ if its tail function $\bar{F}$ meets this criterion with exponents that can be chosen to satisfy $\alpha_{0} < \alpha \le \beta < \beta_{0}$.

    \item[\textbf{Quasi-Decreasing Hazard Rate ($\mathsf{QDHR}$)}] This class relaxes the strict condition of a monotonically decreasing hazard rate. A distribution belongs to $\mathsf{QDHR}$ if its hazard rate $h(x)$, for all $x \ge x_{0}$, is bounded by a strictly increasing function $m$. The relationship is formalized by an exponent $\gamma \ge 1$ and constants $C_{0}, x_{0} > 0$ in the following manner:
    $$
    m(x) \le \frac{1}{h(x)} \le m(C_{0}x^{\gamma})
    $$

    \item[\textbf{Gumbel Domain of Attraction ($\mathsf{MDA}(\Lambda)$)}] Rooted in extreme value theory, this class contains distributions whose tails are lighter than any Pareto-like distribution, meaning they decay faster than any power law. For the purpose of our analysis, the essential property of a distribution in $\mathsf{MDA}(\Lambda)$ is that its tail function $\bar{F}(x)$ must satisfy:
    $$
    \bar{F}(x) = o(x^{-\alpha}) \quad \text{for all } \alpha > 0.
    $$
\end{description}
\end{definition}

\vfill

\end{document}